\documentclass[conference]{IEEEtran}
\IEEEoverridecommandlockouts
	\usepackage[turkish,shorthands=:!]{babel}
	\usepackage[utf8]{inputenc} % Kullanılan encodinge göre utf8 yerine latin5 de yazılabilir.
	\usepackage[T1]{fontenc}
	\usepackage[bookmarks=false]{hyperref}

\newcommand{\pastF}{\ensuremath{F^-}}
\newcommand{\pastG}{\ensuremath{G^-}}
\newcommand{\since}{\ensuremath{S}}

\newtheorem{example}{\bf Örnek}[section]
\newtheorem{theorem}{\bf Teorem}[section]

\newtheorem{proof}{\bf İspat}[section]

\usepackage{cite}
\ifCLASSINFOpdf
  \usepackage[pdftex]{graphicx}
\else
\fi
\usepackage{color}
\usepackage{amsfonts}
\usepackage{algorithm,algorithmicx,algpseudocode}
\usepackage{hyperref}
\usepackage{mathtools}
\usepackage{graphicx}
\usepackage{amsmath}

\usepackage[utf8]{inputenc}
\usepackage{multirow}
\usepackage{array}
\usepackage[lofdepth,lotdepth]{subfig}
\usepackage{tikz-qtree}
\usepackage{tikz}
\usetikzlibrary{decorations.shapes}
\usetikzlibrary{shapes, arrows}
\tikzstyle{line}=[draw]
\tikzstyle{vecArrow} = [thick, decoration={markings,mark=at position
   1 with {\arrow[semithick]{open triangle 60}}},
   double distance=1.4pt, shorten >= 5.5pt,
   preaction = {decorate},
   postaction = {draw,line width=1.4pt, white,shorten >= 4.5pt}]

\AtBeginDocument{%
  \renewcommand\tablename{TABLO}
}

\AtBeginDocument{%
  \renewcommand\abstractname{Abstract}
}

\newcommand\copyrighttext{%
  \footnotesize \textcopyright 2020 IEEE. Personal use of this material is permitted.
  Permission from IEEE must be obtained for all other uses, in any current or future
  media, including reprinting/republishing this material for advertising or promotional
  purposes, creating new collective works, for resale or redistribution to servers or
  lists, or reuse of any copyrighted component of this work in other works.
  }
\newcommand\copyrightnotice{%
\begin{tikzpicture}[remember picture,overlay]
\node[anchor=south,yshift=10pt] at (current page.south) {\fbox{\parbox{\dimexpr\textwidth-\fboxsep-\fboxrule\relax}{\copyrighttext}}};
\end{tikzpicture}%
}

\begin{document}

\title{Kontrol Edilebilir ptSTL Formülü Sentezi \\
Synthesis of Controllable ptSTL Formulas}
\author{\IEEEauthorblockN{\textit{Irmak Saglam ve Ebru Aydin Gol}}
\IEEEauthorblockA{
Bilgisayar Mühendisliği, ODTÜ\\
Ankara, Türkiye\\
saglam.irmak@metu.edu.tr, ebrugol@metu.edu.tr\\
}
}

\maketitle
\copyrightnotice

\begin{ozet}
Bu bildiride anomali tespiti ve önlenmesi problemine, Sinyal Zamansal Mantığı (Signal Temporal Logic) tabanlı iki aşamalı bir çözüm sunulmaktadır. İlk aşama nedenlerin tespiti, ikinci aşama ise bir kontrol stratejisi ile nedenlerin sistem üzerinde engellenmesidir. İki aşama birbirine bağımlıdır. Bu bildiride, ilk aşama olan istenmeyen olayların nedenlerinin tespitinde kullanılan neden formülü şablonu geliştirilmektedir. Bildiride kullanılan şablon ile bütün kontrol edilebilir formüller tanımlanabilmektedir. Bu şablon için verimli bir formül sentezleme algoritması sunulmuş, ve sonuçlar örnek bir sistem üzerinde gösterilmiştir. 
%Bildiride sunulan şablon, geçmişte kullanılmış şablonları kapsamaktadır ve onlardan daha açıklayıcıdır. Hatta, bu şablonla üretilen formüllerin kümesinden daha açıklayıcı bir şablon kümesi yoktur. Sunulan şablonun kullanımının çözümün hassasiyetini ve başarısını artırması beklenmektedir.

%\boldmath
\end{ozet}
\begin{IEEEanahtar}
Formel metotlar, STL, gereksinim türetme, anomali tespiti ve önlenmesi.
\end{IEEEanahtar}
\begin{abstract}

In this work, we develop an approach to anomaly detection and prevention problem using Signal Temporal Logic (STL). This approach consists of two steps: detection of the causes of the anomalities as STL formulas and prevention of the satisfaction of the formula via controller synthesis. This work focuses on the first step and proposes a formula template such that any controllable cause can be represented in this template. An efficient algorithm to synthesize formulas in this template is presented. Finally, the results are shown on an example.

 % In this work, we are improving the formula template used at the first step. The new formula type that is presented has higher explanatoriness level than the existing ones. Furthermore, all the possible controllable cause formulas can be detected by the template. As a result, the overall precision and success of the approach is bound to improve.

\end{abstract}



% \fi
\IEEEpeerreviewmaketitle

\IEEEpubidadjcol

\section{G{\footnotesize İ}r{\footnotesize İ}ş}

Sistem izleri üzerinde istenmeyen olayların işaretlenip, bu olayların oluşmasını engelleyecek kontrol stratejilerinin üretilmesi robotik alanında karşılaşılan bir problemdir. İzlerin analiz edilip, istenmeyen olayların önlenmesine yönelik strateji geliştirilmesi \textit{anomali tespiti ve önlenmesi} problemi olarak da bilinir. Bu problem veri bilimi alanında farklı açılardan ele alınmaktadır~\cite{7424283}. Diğer taraftan bu problemi hedef alan popüler yaklaşımlardan biri, formel metod yaklaşımlarıdır.

Formel kontrol uygulamalarında amaç, zengin bir gereksinim dilinde ifade edilmiş olan bir formülden, bir sistem için kontrol stratejisi üretmektir~\cite{Kavraki:MPlanning,tabuada2006linear}. Formel metotlarda dinamik sistemler için kullanılan ana çözüm stratejisi, sistemin sonlu bir soyut modelini çıkarmaktır~\cite{tabuada2006linear}. Bu sonlu model için otomata teorik yaklaşımlar ile verilen gereksinimden bir kontrol stratejisi üretilir, ardından bu strateji dinamik sisteme uyarlanır. Karmaşık gereksinimler için otomatik olarak kontrol stratejisi üretilmesini sağlayan bu yaklaşım robotik uygulamaları için kullanılmaktadır~\cite{Kavraki:MPlanning}. Gerekirci olmayan (non-deterministic) ve stokastik sistemler için de benzer yaklaşımlar bulunmaktadır~\cite{LaWaAnBe-ICRA10}. Ancak soyut modelin oluşturulmasında kullanılan metotlar politopik işlemler içermekte ve bu da ölçeklendirilmelerini (çok boyutlu sistemlere uygulanmalarını) imkansız kılmaktadır. Ek olarak bu metotlar sistem dinamiğine özgü işlemler içermekte (örn. doğrusal) ve karmaşık dinamiklere sahip sistemler üzerinde uygulanamamaktadırlar. 

Bu bildiride istenmeyen olayların engellenmesi için kontrolcü üretimi problemine iki aşamalı bir çözüm sunulmaktadır. İlk aşamada etiketli sistem izlerinden oluşan bir veri kümesi kullanılarak, etiketlenmiş olayları ifade edecek belirli bir yapıdaki geçmiş zamanlı sinyal zamansal mantık (past time Signal Temporal Logic - ptSTL) formülleri üretilir. İkinci aşamada ise, bulunan formüllerin sağlanması, yani istenmeyen olayların oluşması, bir kontrolcü ile engellenir. Bu bildiride \textit{engellenebilir} veya \textit{kontrol edilebilir} formüllerin üretilmesine odaklanılmış ve üretilen formüllerin kontrol edilebilir olduğu kanıtlanmıştır. Bu ispat kontrolcü üretilmesi üzerine bir metot sunsa da, daha verimli kontrolcü üretme problemi ileriki çalışmalarda ele alınacaktır. 

Sunulan çözüm yöntemi, sistem dinamiğinden bağımsızdır. Temel olarak, sistem dinamiğine bakılmadan sadece sistem izlerinden bir formül türetilmekte ve bu formül için sistem dinamiğinden bağımsız geri beslemeli bir kontrolcü üretilmektedir. Bu sayede geliştirilen metot farklı özellikleri olan dinamik sistemlere uygulanabilecektir. Ek olarak, formül türetme ve kontrolcü sentezleme adımlarında karmaşık işlemler yer almadığı için çok boyutlu sistemlere de uygulanabilmektedir.  
 
Etiketli izlerden formül türetme problemi literatürde farklı açılardan ele alınmıştır~\cite{Yoo-RSS-17,mining_journal,Bartocci2014,CDC2019}. Bu çalışmaların bir çoğunda hedef bir sınıflandırıcı üretmektir ve kontrol edilebilme özelliği değerlendirilmiştir. \cite{CDC2019} numaralı referansta, bu çalışmada olduğu gibi, kontrol edilebilir bir formül türetilmiştir. Ancak \cite{CDC2019} çalışmasında ele alınan formül şablonu, sadece belirli bir süre boyunca aynı kontrol girdisinin uygulandığı durumları tanımlyabilmektedir. Bu nedenle de olası diğer kontrol edilebilir nedenleri ifade etmekte yetersiz kalmaktadır. 

%Ancak, (self-citation)'da verilen çözümün ilk aşamasında kullanılan formül şablonu, olası bütün nedenleri ifade etmekte yetersiz kalmaktadır. 
% Bu nedenle, etiketli olayların nedenlerinin bir kısmı, gereğinden daha az hassas formüller tarafından ifade edilmeye çalışılırlar. Bu da, ikinci aşamada üretilen kontrolcünün sistemdeki istenmeyen etiketli olaylar için yeni nedenler üretme ihtimalini artırır. Sonuç olarak, ilk aşamadaki formül şablonunun yetersizliği, kontrolcünün istenmeyen olayların sayısını azaltmadaki başarısını düşürür.

Bu çalışmada, yeni bir kontrol edilebilir formül şablonu tanımlanmış, ve bu şablona uygun formüllerin üretilmesi için bir formül sentezleme algoritması anlatılmıştır. Kontrol edilebilir bütün formüllerin bu şablon ile ifade edilebileceği gösterilmiştir. Bu sayede, istenmeyen olayların etiketli olduğu bir kontrol sisteminde, en hassas kontrol edilebilir nedenlerin dahi bu şablonla üretilebileceği garanti edilmektedir. 

% (self-citation)'daki çözüm metodunun ilk aşamasında kullanılması için yeni bir kontrol edilebilir neden formülü şablonu öneriyoruz. Önerdiğimiz formül şablonu herhangi bir olası kontrol formülünü dışarıda bırakmamakta, dolayısıyla, parametreler yeterince büyük tutulursa, istenmeyen olayların etiketli olduğu bir kontrol sisteminde, en hassas kontrol edilebilir nedenlerin dahi bu şablonla üretilebileceğini garanti etmektedir. Bu garanti, sonuç olarak, ikinci aşamada benimsenecek bir kontrol stratejisinin, yeni nedenler üretmekten kaçınmada en iyi performansı göstereceğini ima etmektedir.\\

%%% SONRADAN EKLENEBİLİR, YER KALIRSA
%Bu çalışmada, ikinci kısımda ön bilgiler ve problem tanımı vererek başlıyoruz. Üçüncü kısımda kontrol edilebilirliği tanımlıyoruz. IV'te yeni formül şablonumuzu tanıtıyoruz ve bu şablondaki formüllerin kontrol edilebilirliklerini tespit etmeyi algoritmik olarak kolaylaştıracak bir teorem kanıtlıyor, daha sonra kontrol edilebilir tüm formüllerin bu şablondaki bir formüle denk olduğunu kanıtlıyoruz. Böylece, önerdiğimiz şablonla bulacağımız kontrol formülleri kümesinin, hiçbir olası kontrol formülünü dışarıda bırakmayacağını kanıtlıyoruz. V'te etiketli bir sinyalde, tanıttığımız şablondaki kontrol edilebilir formülleri tespit etmek için kullandığımız algoritmayı tanıtıyor, VI'da sonuçlarımızı bir örnek üzerinde göstererek makaleyi bitiriyoruz.

\section{Ön B{\footnotesize İ}lg{\footnotesize İ}ler ve Problem Tanımı}

\subsection{Sistem tanımı}

Kesikli bir dinamik sistem Denklem~\eqref{eq:system}'deki gibi ifade edilmiştir. 
%Aşağıdaki kuralla tanımlanmış parçalı sabit bir kontrol sistemi üzerinde,
\begin{equation}\label{eq:system}
x_{k+1} = f(x_k, u_k, w_k)
\end{equation}
$x_k = [x_k^{0}, \ldots, x_k^{n-1}] \in \mathbb{X} \subset \mathbb{R}^n$ sistemin $k$ anındaki durumunu, $u_k = [u_k^{0}, \ldots, u_k^{m-1}] \in \mathbb{U} \subset \mathbb{R}^m$ kontrol girdilerini, $w_k \in \mathbb{W} \subset \mathbb{R}^l$ ise sistem gürültüsünü temsil eder. Bu çalışmada, her kontrol girdisinin değerlerini sınırlı bir kümeden aldığı varsayılmıştır. Yani $\mathbb{U}^i = \{c^{(i,1)},\ldots, c^{(i,M_i)}\} \subset \mathbb{R}, i=0,\ldots,m-1$ ve $\mathbb{U} = \mathbb{U}^0 \times \ldots \times \mathbb{U}^{m-1}$. Sistemin sınırlı uzunluktaki bir izi \[\mathbf{x} = (x_0, u_0), \ldots, (x_{N}, u_{N}),\] ile gösterilir, öyle ki bu iz üzerindeki her bir $k$ anı için $x_{k+1} = f(x_k, u_k, w_k)$ eşitliğini sağlayan bir gürültü değeri $w_k \in \mathbb{W}$ vardır. Bir $\mathbf{x}$ izinin etiketi, aynı uzunlukta ikili (binary) bir seridir ve aşağıdaki gibi gösterilir,
\begin{equation}\label{eq:labels}
\mathbf{l} = l_0, \ldots, l_N,   \quad l_k \in \{0,1\}
\end{equation}
$l_k = 1$ etiketi $k$ anında istenmeyen bir olay gerçekleştiğini ifade eder. %Etiket dizisi sistem ve kontrol girdileri üzerini tanımlanmış bir $g$ fonksiyonu tarafından üretilebilir. $g: \mathbb{X} \times \mathbb{U} \to \{0,1\}$, ve $l_k = g(x_k, u_k)$.
Sistem ~\eqref{eq:system} için etiketli bir iz kümesi aşağıdaki gibi gösterilir: 
\begin{equation}\label{eq:dataset}
\mathcal{D} = \{(\mathbf{x} _i, \mathbf{l} _i)\}_{i=1, \ldots, D}.
\end{equation}

% Bir $\mathbf{x}$ sisteminde, bir $k$ anında $\phi$ formülünün sağlandığını anlatmak için $(\mathbf{x},k)\models \phi$ gösterimini kullanıyoruz. Bu gösterim, $\phi$ formülündeki $x^i$ ve $u^i$ girdileri $\mathbf{x}$'in $k$ anındaki değerleri ile girildiğinde, formülün $\mathbf{T}$(Boole sabiti - doğru) değerini verdiğini gösterir. Benzer şekilde $(\mathbf{x},k) \not \models \phi$, sistemin $k$ anında $\phi$ formülünü sağlamadığını, yani, $x^i$ ve $u^i$ girdileri sistemdeki şekliyle formüle girildiğinde, formülün $\mathbf{F}$(Boole sabiti - yanlış) değerini verdiğini gösterir.
\subsection{Sinyal Zamansal Mantık}

Geçmiş zamanlı sinyal zamansal mantık için sözdizim kuralları
\begin{align*}
& \varphi = \mathbf{T} | x>c | x < c | u = c | \neg\varphi | \varphi_1 \wedge \varphi_2 | \varphi_1 \vee \varphi_2 | \\ & \quad \quad \quad \quad \varphi_1 \since_{[a,b]} \varphi_2 | \pastF_{[a,b]} \varphi | \pastG_{[a,b]}  \varphi
\end{align*}
şeklinde tanımlanır. $\neg$ - değil, $\wedge$-ve ile $\vee$-veya mantık operatörleridir. $\since_{[a,b]}$(olduğundan beri, since), $\pastF_{[a,b]}$(geçmişte bir zaman, previously) ve $\pastG_{[a,b]}$(şimdiye kadar, always) ise geçmiş zaman operatörleridir. Bir ptSTL formülünün semantiği bir sinyal $\mathbf{x}$ ve zaman $k$ için tanımlanır. Özetle, $k$ anında, $ \pastF_{[a,b]} \varphi$ formülünün sağlanması (yani, $(\mathbf{x},k) \models \pastF_{[a,b]}  \varphi$), $\varphi$'nin $[k-a, k-b]$ aralığında en az bir an sağlandığını, $(\mathbf{x},k)\models \pastG_{[a,b]} \varphi$ , $\varphi$ formülünün $[k-a, k-b]$ aralığındaki her anda sağlandığını, $(\mathbf{x},k) \models \varphi_1 \since_{[a,b]} \varphi_2$ formülü ise $\varphi_2$ formülünün $[k-a, k-b]$ aralığında bir $t$ anında sağlandığını ve $\varphi_1$ formülünün $[t, k-b]$ aralığındaki her anda sağlandığını ifade eder. $(\mathbf{x},k) \models \varphi$ notasyonu $\mathbf{x}$ sinyali üzerinde, $k$ anında $\varphi$ formülünün sağlandığını ifade eder. 

\newtheorem{defn}{\textbf{Tanım}} 
\begin{defn}[Operatör Sayısı]\label{defn:definition1} 
Bir formülün \textit{Operatör Sayısı}, o formülün içindeki Boole ve STL operatörlerinin ($\{\vee,\wedge,\neg,\pastF,\pastG,\since\}$) kaç tane olduğunu ifade eder. Bir $\phi$ formülünün operatör sayısı $os(\phi)$ olarak gösterilir. Örneğin, $os( (x^1 < 10) \wedge (u^0 = 2) ) = 1$, $os(\phi_1 \since_{[a,b]} \phi_2) = os(\phi_1) + os(\phi_2) + 1$.
\end{defn}

\begin{defn}[$X$ Operatörü]\label{defn:definition2} 
$X$, 'bir önceki anda' operatörüdür, formülün zamansal olarak bir önceki anki sağlanma değerinden bahsetmeyi sağlar. Herhangi bir $\mathbf{x}$ izi ve ptSTL $\phi$ formülü için:
$(\mathbf{x}, k) \models X \phi \iff (\mathbf{x}, k-1) \models \phi $. $X \phi$ formülü $\pastF_{[1,1]} \phi$ ile $\pastG_{[1,1]} \phi$ formüllerine denk olduğu için standard notasyonda bulunmamaktadır. 
\end{defn}

 $n$ tane $X$ operatörünün ard arda sıralandığı $XX\ldots X \phi$ formülü, kısaca $X^n \phi$ olarak gösterilir.

\textbf{Not:} Operatör sayısı, yalnızca Boole ve STL operatörleri üzerinde tanımlanmıştır. Bir formülün içindeki $X$ operatörü sayısı, o formülün \textit{Operatör Sayısı}'nı etkilemez.

\subsection{Problem Tanımı}

% Bu veri kümesi ışığında amaç, ettle asgari sayıda çelişen bir sonuç elde edecek şekilde bir ptSTL formülü tanımlayabilmektir.

Bu çalışmada amaç bir dinamik sistem~\eqref{eq:system} ve bu sistemin etiketlenmiş izlerinden oluşan veri kümesi kullanılarak, bu etiketleri ifade edecek \textbf{kontrol edilebilir} bir ptSTL formülü bulmaktır. Bu problem zaman serisi verileri üzerinde ikili bir sınıflandırma problemi olarak değerlendirilebilir. Bu çalışmada, üretilen formülün veri kümesi ile uyumluluğu  $F_{\beta}$ skoru ile ölçülmüştür.

%yukarıda tanımlanan sistemin etiketli sinyal kümesinden, bu sinyallerin etiketlerini yüksek başarıyla ifade eden  Burada yüksek başarıdan kasıt, bulunan formülün veri kümesi üzerinde çalıştırılmasıyla türetilen etiketlerin, sinyal kümesindeki etiketlerle farkının asgari sayıda olmasıdır. Asıl veri kümesinde "0" olarak etiketlenip, üretilen ptSTL formülünün "1" olarak yansıttığı her sinyal parçası ve asıl veri kümesinde "1" olarak etiketlenip, üretilen ptSTL formülünün "0" olarak işaretlediği her sinyal parçası uyumsuzluk sayısını yükseltir. 
% Asgari uyumsuzluk azami başarı anlamına gelir. Bir formülün uyumunu puanlamak bir sınıflandırma problemidir ve biz bu çalışmada, uyumluluk değerini 

% [referans] Örneklerimizde $\beta = 0.3$ değerini kullanıyoruz.

\begin{example}\label{ex:trafficsystem}
Bildiride, Şekil ~\ref{fig:traffic}'de verilen, 5 yol ve 2 trafik ışığından oluşan trafik sistemi, örnek olarak kullanılmıştır. Her yolun sonlu bir kuyruk olarak modellenmesi ile parçalı doğrusal bir dinamik sistem modeli tanımlanır~\cite{coogan2016traffic}. Bu modelde sistem durumları $x^0, x^1, x^2, x^3$ ve $x^4$, sırasıyla $0,1,2,3$ ve $4.$ yollardaki araç sayılarını temsil eder.  Ana yollar, $i \in \{0,1,2\}$, için $x^i \in [0,40]$ ve yan yollar, $i \in \{3,4\}$, için  $x^i \in [0,20]$ olarak tanımlanmıştır. Kontrol girdileri trafik ışıklarıdır. $u^0$, $s0$; $u^1$, $s1$ ışığını temsil eder. $u^0, u^1 \in \{0,1\}$. Bir trafik ışığının $0$ değeri alması yatay trafik akışına izin verip dikey akışı durdurduğunu, $1$ değeri alması ise tam tersini ifade eder. 

Bu sistem için etiketli veri kümesi oluşturulurken, sistem rastgele başlangıç koşullarından başlatılmış ve simulasyon sırasında, her adımda kontrol girdileri rastgele seçilmiştir. Bir adım sonrasında, herhangi bir yoldaki araç sayısı, o yolun kapasitesinin $\%75$'inden fazlaysa, yani $k+1$ anındaki sistem durumu, $\mathbf{x}_{k+1}$, \eqref{eq:violationformula} formülünü sağlamıyorsa, $l_k = 1 $ olarak atanmıştır (bir adım öncesindeki etiket erken tanı için kullanılmıştır).
\begin{equation}\label{eq:violationformula}
x^0 < 30 \wedge x^1 < 30  \wedge  x^2 < 30 \wedge x^3 < 15 \wedge x^4 < 15 
\end{equation}
Veri kümesi oluşturulurken uzunlukları $100$ olan toplam  $20$ simulasyon yapılmıştır. Bu veri kümesinde, \eqref{eq:violationformula} formülünün sağlanmama oranı (trafik tıkanıklığı oranı) $\%46$ olmuştur. Bu örnek sistem için amaç, tıkanıklığa yol açan durumları kontrol edilebilir formüller halinde üretmek, ardından bu formüller ile sinyal kontrol stratejisi üretip trafik tıklanıklığını azaltmaktır. 
\end{example}

\begin{figure}[h]
\centering
\includegraphics[width=7cm]{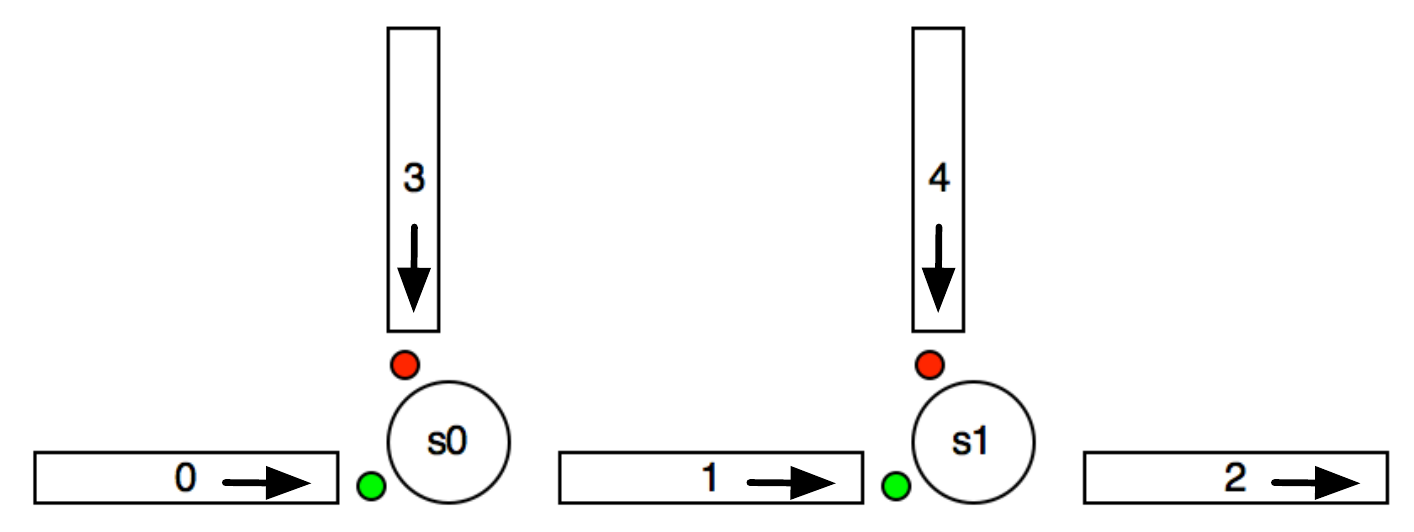}
\caption{5 yol ve 2 trafik ışığından oluşan bir trafik sistemi.}\label{fig:traffic}
\end{figure}

\section{Kontrol Ed{\footnotesize İ}leb{\footnotesize İ}l{\footnotesize İ}rl{\footnotesize İ}k}

Bir kontrol sistemi~\eqref{eq:system} için, kontrol girdileri ($u^i \in \mathbb{U}$) ve sistem değişkenleri ($x^i \in \mathbb{X}$) üzerine tanımlanmış bir ptSTL formülü $\Phi$, eğer her bir $k$ anında sistemin bir $k'<k$ anına kadar olan izi,
\[(x_0,u_0), \ldots,(x_{k'-1},u_{k'-1}, x_{k'})\]
için sistemin $k$ anında $\Phi$
formülünü sağlamamasını garanti edecek ($(\mathbf{x},k) \not \models \Phi$) bir kontrol girdisi dizisi $u_{k'},u_{k'+1}, \ldots u_{k}$ bulunabiliyor ise, $\Phi$ formülü \textit{kontrol edilebilirdir} denir. Diğer bir deyişle, $\Phi$'nin kontrol edilebilirliği, $\neg \Phi$'nin sağlanabilirliğine (satisfiability) denktir.  Verilen kontrol edilebilirlik tanımı, değerlendirilen sistem için, her bir $k$ anında sistemin $\Phi$ formülünü ihlal etmesini sağlayacak bir kontrol stratejisinin olduğunu garanti eder. 

% için eğer herhangi bir sistem durumu dizisi $x_0, \ldots, x_k$ için, $k$ anında sistem izinin $\Phi$ formülünü sağlamamasını garanti edecek bir kontrol girdisi dizisi $u_0, \ldots, u_k$ bulunabiliyor ise, ($(\mathbf{x},k) \not \models \Phi$), $\Phi$ formülü \textit{kontrol edilebilirdir} denir. 

% ESKI VERSİYON
% Kontrol girdileri ($u^i \in \mathbb{U}$) ve sistem değişkenleri ($x^i \in \mathbb{X}$) üzerine tanımlanmış bir ptSTL formülü $\Phi$ için eğer herhangi bir sistem durumu dizisi $x_0, \ldots, x_k$ için, $k$ anında sistem izinin $\Phi$ formülünü sağlamamasını garanti edecek bir kontrol girdisi dizisi $u_0, \ldots, u_k$ bulunabiliyor ise, ($(\mathbf{x},k) \not \models \Phi$), $\Phi$ formülü \textit{kontrol edilebilirdir} denir. 

%$\mathbf{x}: (x_0, u_0), \ldots (x_{k-1}, u_{k-1})$; $k$ anında kontrol girdilerinin ($u^i_k$) değerleri seçilerek elde edilen sistem izinin $\Phi$ formülünü sağlamaması garanti edilebiliyor ise ($\mathbf{x}, k \not \models \Phi$), 

% Sistemde, istenmeyen olayları açıklayan formüllerden özellikle kontrol edilebilir olanlarını bulmak isteme nedenimiz, bir kontrolcü ile bu formüllerin kontrol girdilerini değiştirerek bu nedenlerin sistemde bir daha hiç sağlanmamasını garanti edebilecek olmamızdır. \\

Örnek bir kontrol edilebilir formül $u^1$'in değer kümesi $\mathbb{U}^1=\{0,1\}$ olduğu durumda, $\Phi = (u^1=1) \wedge (x^1>30)$'dur. $x^1$'in $30$'dan küçük veya büyük olmasından bağımsız olarak, $u^1$'e $0$ değeri atanarak, herhangi bir anda $(\mathbf{x},k)\not \models \Phi$ durumu sağlanabilir. % Başka bir örnek, yalnızca kontrol girdilerinden oluşan  $\mathbb{U}^1=\mathbb{U}^2=\{0,1\}$ için, $\Phi = u^1=1 \vee u^2=0$'dır. Bir kontrol edilemez formül ise, aynı kontrol girdisi değer kümeleri için $\Phi = (u^1=0 \vee u^1=1)$'dir. Yalnızca kontrol girdilerinden oluşmasına rağmen, $u^1$'in alabileceği $0$ ve $1$'den başka değer olmadığından, bu formül bir totolojidir, hiçbir olası kontrol girdisi kümesi için $(\mathbf{x},k)\not \models \Phi$ sağlanamaz. 

% EBRU: FALSE sayılsın, önemli değil. 
%Bir formüle \textit{kontrol edilebilir neden formülü} denmesi için, o formülün en az bir sistemde, bir etiketli izi açıklayabilmesi gerekir. Yani, her kontrol edilebilir neden formülü $\Phi$ için, $\exists \mathbf{x}, k$ öyle ki $\mathbf{x},k\models\Phi$.

%\section{Kontrol Edilebilir Formül Şablonu}
% [] Referanstaki çalışmada kontrol edilebilir formül şablonu $\Psi$ aşağıdaki gibi verilmiştir.
% \begin{equation}\label{eq:formulashape}
% \Psi :=  \Phi_1 \vee \ldots  \vee \Phi_p,  
% \end{equation} 
% \quad \quad \quad \quad \quad \quad \quad \quad \quad \quad \quad öyle ki
% \begin{equation}\label{eq:formulainside}
%   \quad \Phi_i  := (\pastG_{[1,b_i]} u^j=c_i) \wedge ( \pastF_{[1,1]} \phi_{i} ) , 
 %  \end{equation}
 %  ve $\phi_i$, $\{ x^0, \ldots, x^{n-1}\} \cup \{u^0,  \ldots, u^{m-1}\}$ üzerinde herhangi bir ptSTL formülü.\\

% Bizim bu çalışmada önerdiğimiz kontrol formülü şablonunda ise \eqref{eq:formulashape}'teki $\vee$'lerle ayrılmış çoklu neden formülü yapısı korunmuş, $\Psi$'yı oluşturan alt neden formülleri $\Phi_i$'ların her birinin şablonu ise aşağıdaki gibi tanımlanmıştır:
Bu çalışmada önerilen kontrol edilebilir formül şablonu $\Psi$~\eqref{eq:formulashape} olarak tanımlanmıştır.  
\begin{equation}\label{eq:formulashape}
\Psi :=  \Phi_1 \vee \ldots  \vee \Phi_p,  
\end{equation} 
$\Psi$~\eqref{eq:formulashape} bir çok kontrol edilebilir $\Phi_i$~\eqref{eq:neden} \textit{neden} formülünün \textit{veya} operatörü ile birleştirilmesi ile oluşturulmuştur. Her bir neden formülü~\eqref{eq:neden}'da verilen şablona uygun olarak tanımlanmıştır. 

\begin{equation}\label{eq:neden}
\Phi^{(l,r)}_i = \phi^{\{u\}}_i \wedge  \phi^{\{x,u\}}_i
\end{equation}
Burada $\phi^{\{u\}}_i$ yalnızca kontrol girdileri ($u$) içeren bir ptSTL formülüdür, $\phi^{\{x,u\}}_i$ ise hem kontrol girdilerini hem sistem durumunu ($x,u$) içeren bir ptSTL formülüdür ve  $os(\phi^{\{u\}}) = l$ ve  $os(\phi^{\{x,u\}}) = r$'dir. 

\begin{theorem}\label{theorem:Teorem1} 
$\phi^{\{u\}}$ kontrol edilebilirse, $\Phi^{(l,r)} = \phi^{\{u\}} \wedge  \phi^{\{x,u\}}$ kontrol edilebilirdir.
\end{theorem}
\begin{proof}[İspat]
$\phi^{\{u\}}$ formülünün kontrol edilebilir olması, her $k$ anında, her sistem izi için $(\mathbf{x}, k) \not \models \phi^{\{u\}}$'ni sağlayacak kontrol $u$ bulunabileceğine denktir. $(\mathbf{x}, k) \not \models \phi^{\{u\}} \implies (\mathbf{x}, k) \not \models \phi^{\{u\}}\wedge  \phi^{\{x,u\}}=\Phi^{(l,r)} $ gözlemi ispatı tamamlar. 
% Sistemimin $m$ kontrol, $n$ sistem girdisi olduğunu varsayalım. $\phi^{\{u\}}$ kontrol edilebilirse, her $k$ anında, olası her sistem girdisi dizisi $x_k = (x_k^1, \ldots, x_k^n)$ için, $u_k^i=c^i \in \mathbb{U}\quad \forall i \in \{1, \ldots, m\}$ olacak şekilde, kontrol girdilerinin alabildiği en az bir değerler dizisi $(c^1, ..., c^m)$ vardır ki $(\mathbf{x},k) = ((x, (c^1, \ldots, c^m) ), k) \not \models \phi^{\{u\}} \implies (\mathbf{x},k) \not \models \Phi^{\{u\}} \wedge  \phi^{\{x,u\}}=\Phi^{(l,r)}$. Bu nedenle, aynı kontrol girdisi değer kümesi $(c^1, ..., c^m)$ ile $ \Phi^{(l,r)}$ da kontrol edilebilirdir.
\end{proof}
\begin{theorem}\label{theorem:Teorem2} 
Denklem~\eqref{eq:system}'deki gibi tanımlanmış sistemler için, $\textbf{F}$ (false) formülü duşındaki kontrol edilebilir bütün formüller~\eqref{eq:neden} şeklinde yazılabilir. 
% $\Phi^{(l,r)} = \phi^{\{u\}} \wedge  \phi^{\{x,u\}}$ şablonunda yazılabilir.
\end{theorem}
\par \textit{İspat Fikri: } Bildirideki yer kısıtından ötürü, ispat yerine teoremin ispatında kullanılan temel fikir anlatılmıştır. 

$\Phi^{(l,r)}$ şablonunun herhangi bir ptSTL formülüne getirdiği tek kısıt, formülde, formülün geri kalanına "ve" operatörü ile bağlı, yalnızca kontrol girdilerinden oluşan bir birim bulunmasıdır.
Verilen kontrol edilebilir formül tanımından ötürü, $\textbf{F}$ formülüne denk olan formüller dışındaki bütün formüllerde kontrol girdisi bulunması gereklidir. Örneğin $\varphi$ formülü hiçbir kontrol girdisi içermeyen rastgele bir "kontrol edilebilir" formül olsun. Buna göre formülün içerdiği bütün eşitsizlikler $x^i \sim c$ formundadır ($\sim \in \{<,>\}$). Bu kontrol edilebilirlik tanımı ile çelişmektedir, çünkü~\eqref{eq:system}'deki gibi tanımlanan herhangi bir sistem için bu formülün ihlal edilmesini sağlayacak kontrol dizisi bulunması garanti edilemez. Ulaşılan çelişki, $\varphi$'nin en az bir kontrol girdisi bilgisi ($u^i \sim c$) taşıdığını gösterir.
 
Herhangi bir ptSTL formülü üzerinde, aşağıdaki dönüşüm kuralları uygulanarak $\pastG_{[a,b]}$ ve $\pastF_{[a,b]}$ operatörleri, $X$ operatörleri ile değiştirilebilir. 
\begin{align*}
& \pastG_{[a,b]}\phi = X^b\phi \wedge X^{b+1} \wedge \ldots \wedge X^a \phi \\
& \pastF_{[a,b]}\phi = X^b\phi \vee X^{b+1} \vee \ldots \vee X^a \phi \\
\end{align*}
$X$ operatörleri, operatörler yalnızca atomik formüllerin önünde kalıncaya dek $\vee$ ve $\wedge$ operatörleri üzerine dağıtılabilir ($\sim \in {\wedge, \vee}$):
\begin{equation*}
X^n(\phi_1 \sim \phi_2) = X^n \phi_1 \sim X^n\phi_2
\end{equation*}
Her değişken eşitsizliğini önündeki $X$ operatörleri ile birlikte bir atomik formül olarak düşünerek, formülün  birleşim normal formunda (CNF-conjunctive normal form) denkliği yazılabilir. Sonrasında, değil $\neg$ operatörü, $X$ operatörlerinin içine aşağıdaki kurala göre geçirilebilir:
\begin{equation*}
\neg X^n \phi = X^n \neg \phi
\end{equation*}
Betimlenen işlemler sonrasında ulaşılan formül aşağıdaki formdadır: 
\begin{equation}\label{eqn1}
\phi_1 \wedge \ldots \wedge \phi_z
\end{equation} 
öyle ki, her bir $\phi_i$ yalnıca $X^n x^j \sim c$ ve $X^n u^j \sim c$ formundaki atomik formüllerin $\vee$ operatörü ile birleşiminden oluşur. İspatın başında sunulan argümandan ötürü, \eqref{eqn1} formülünün ihlal edilebilirliği; ancak bu $\phi_i$ formüllerinden en azından bir tanesi sadece $X^n u^j \sim c$ formundaki formüllerden oluşuyorsa garanti edilebilir. Bu formülün diğerleri ile $\wedge$ operatörü ile bağlanması, ve herhangi bir kontrol edilebilir formülün bu forma dönüştürülebilmesi sunulan ispatı tamamlar.

\par\textbf{Not:}Kısıtlı zaman limitlerine sahip $\,\since$ operatörü $\pastF$ ve $\pastG$ operatörleri tarafından ifade edilebildiğinden, yer kısıtı nedeniyle $\since$ operatörü kanıta dahil edilmemiştir. % Yine kanıtın seyrini değiştirmeyen bir basitleştirme yaparak $x<c$ ve $x>c$ formüllerini birbirlerinin değili kabul ettik. Bu kabul, algoritma karmaşıklığını azaltmak amacıyla yapılmıştır.\\ \newline

Kontrol edilebilir neden formüllerinin \eqref{eq:neden} şeklinde tanımlanabildiği gösterildi. \ref{theorem:Teorem1}'e göre,  eğer $\phi^{\{u\}}$ kontrol edilebilir ise, $\Phi^{(l,r)}$ de kontrol edilebilirdir.
Bu çalışmada, $\phi^{\{u\}}$'nun kontrol edilebilirliğini tam kapsamlı arama (exhaustive search) ile kontrol edilmektedir.
Yani, $\phi^{\{u\}}$ formülünde geçen tüm girdilerin tüm değerlerinin tüm olası kombinasyonlarını formül üzerinde denenmekte. Eğer  $\phi^{\{u\}}$'nun değerini $\mathbf{F}$ yapan bir kontrol girdisi değeri bulunabilirse,  $\phi^{\{u\}}$ kontrol edilebilirdir sonucuna ulaşılıyor. 

\section{Kontrol Ed{\footnotesize İ}leb{\footnotesize İ}l{\footnotesize İ}r Formül Üretme Algor{\footnotesize İ}tması}
% \textbf{Not:} Kullanılan sistem tanımında \eqref{eq:system}, $k$ anındaki bir izin etiketi, yalnızca $k-1$ anı ve öncesindeki sistem ve kontrol değerlerine bağlı olduğundan, alt neden formülü şablonu \eqref{eq:neden} yerine 
% \begin{equation}\label{eq:pastneden}
% \Phi^{(l,r)}_i = \pastF_{[1,1]}(\phi^{\{u\}}_i \wedge  \phi^{\{x,u\}}_i)
% \end{equation}
% şablonunda aranmaktadır.\\ \newline

Bu çalışmada, \cite{CDC2019} numaralı referansta sunulan formül arama algoritması uyarlanarak \eqref{eq:formulashape} formatında ve \eqref{eq:neden} şeklinde tanımlanmış alt nedenler içeren bir formül üretilmesi sağlanmıştır. 
Bu uyarlama sırasında, operatör sayısının artışı, alt neden formülü yapısındaki değişikliğe göre adapte edilmiştir. Temel olarak \cite{CDC2019} numaralı çalışmada sadece $\pastG_{[0,b]} u = c$ formatındaki kontrol formüllerine izin verilirken (tek operatör $\pastG_{[0,b]}$), bu çalışmada kontrol girdileri üzerine tanımlı herhangi bir kontrol edilebilir formüle izin verilmektedir. Dolayısı ile $\phi^{\{u\}}$ formülünün içerisinde geçen operatörler de yinelemeli algoritma içerisinde hesaplanmıştır.  Ek olarak, formül türetme algoritmasına kontrol edilebilirlik için doğrulama adımı eklenmiştir.

 Bütünlük açısından, detaylara girilmeden formül türetme algoritması, ana hatlarıyla aşağıda verilmiştir.
 
Algoritmanın aldığı parametreler: $\underline{total\_oc}$, $\overline{total\_oc}$ toplam operatör sayısı (\eqref{eq:neden}'deki $l+r$) için alt ve üst limitler, $\overline{p}$: \eqref{eq:formulashape} formülündeki $p$ için üst limit, $\underline{val}$: $\Psi$ formülüne eklenecek her yeni $\Phi^{(l,r)}_i$'nin,$\Psi$ formülünün değerine getirmesi gereken artış için alt limit, $\mathcal{D}$: \eqref{eq:dataset} veri kümesi, $\mathcal{P}$: parametre değer kümesi.\\ \newline
Algoritmada öncelikle
anlık\_os$ = \underline{total\_oc}$,
son\_formül\_başarısı = 0, 
$\Psi$ formülü boş olarak başlatılır.\\ \newline
anlık\_os$  <  \overline{total\_oc}$ \textbf{ve} $\Psi$'nin boyu $\overline{p}$'den küçük iken, aşağıdaki adımlar tekrarlanır:  
\begin{itemize}
\item $l+r$=anlık\_os'yi sağlayan $l,r \in \mathbb{N}$ olacak şekilde tüm $(l,r)$ ikili kombinasyonları bulunur ve \\operatör sayısı $l$ olan tüm $\phi^{\{u\}}$ formülleri şablon olarak üretilerek  $\phi^{\{u\}}$-listesi oluşturulur,  operatör sayısı $r$ olan tüm $\phi^{\{x,u\}}$ formülleri şablon olarak üretilerek $\phi^{\{x,u\}}$-listesi oluşturulur.(Örn, $l=0$ için, trafik örneğinde, $\phi^{\{u\}}$-listesi $= [(u^0 = c_0), (u^1=c_1)]$)
\item  $\phi^{\{u\}}$-listesi'deki formüllerin içinde parametreler yerleştirilir, kontrol edilemeyen formüller elenir, kontrol edilebilir parametrik formüllerden parametrik- $\phi^{\{u\}}$-listesi oluşturulur.
\item parametrik-$\phi^{\{u\}}$-listesi ve $\phi^{\{x,u\}}$-listesi'ndeki formüllerden birer formül alınarak oluşturulabilecek tüm $(\phi^{\{u\}}_i \wedge \phi^{\{x,u\}}_j)$ formülleri oluşturulur ve bu formüllerden $\Phi$-listesi oluşturulur.
\item Şebeke araması (grid-search) ile $\Phi$-listesi'ndeki tüm formül şablonlarına tüm olası parametre değerleri yerleştirilir ve içlerinden en başarılı olanı $\Phi^*$ bulunur.

\item $\Phi^*$'ın başarı değeri $-$ son\_formül\_başarısı $\geq \underline{val}$ ise,

\begin{itemize}
\item $\Psi = \Psi \vee \Phi^*$
\item son\_formül\_başarısı $=$ $\Psi^*$'ın başarısı
\item son\_başarılı\_formül\_os = anlık\_os 
\end{itemize}
\item $\Phi^*$'ın başarı değeri $-$ son\_formül\_başarısı $< \underline{val}$ ise,
\begin{itemize}
\item son\_başarılı\_formül\_os $\neq$ anlık\_os ise,
\begin{itemize}
\item Algoritma sonlandırılır\\
\end{itemize}
\item Aksi takdirde,
\begin{itemize}
\item anlık\_os += 1\\
\end{itemize}
\end{itemize}
\end{itemize}

\textbf{Trafik örneği üzerindeki sonuçlar}: Geliştirilen formül arama algoritması trafik örneği (\ref{ex:trafficsystem}) üzerinde çalıştırıldığında elde edilen sonuçlar, algoritmanın çalıştırıldığı parametre değerleri ile birlikte TABLO \ref{tab:congested}'de verilmiştir:
 
 \begin{table}[h!]
\begin{center}
 \begin{tabular}{|c| c | c | m{6.5cm}| }
 \hline
 $\underline{oc}$ & $\overline{oc}$ & $\overline{p}$ &  $\Psi$  \\ [0.5ex] 
 \hline\hline
 1 & 1 & 1 & $ ( (u^0 = 0) \wedge ((x^3 > 10)\wedge (u^0 = 0)) )$\\
 \hline
 2 & 2 & 1 &$ ((u^0 = 0) \wedge ((x^4 > 10) \wedge (u^1 = 0)) \vee (x^3 > 10))$\\
 \hline
0 & 2 & $\infty$ & $( ( (u^1=0) \wedge (x^4>10)) ) \vee 
                  ( (( u^0 =0) \wedge (x^3 > 10)) ) \vee
                  ( (( u^0 =1) \wedge (x^0 > 30)) ) \vee
                  ( (( u^1 =1) \wedge (x^1 > 30)) )$  \\

\hline\end{tabular}
\end{center}
\caption{}\label{tab:congested}
\end{table}
$\underline{oc}$ ve $\overline{oc}$, sırasıyla $\underline{total\_oc}$ ve $\overline{total\_oc}$'yi temsil etmektedir. İlk iki örnekte $\overline{p}=1$ olduğu için $\underline{val}$ etkisizdir. Son örnekte $\underline{val}=0.001$ değeri kullanılmıştır. $\Psi$, geliştirilen algoritmanın verilen parametrelerle çalıştırılması sonucu elde edilen neden formülüdür.

Sonuçlar Şekil~\ref{fig:traffic} üzerinden incelendiğinde, bulunan neden formüllerinin sezgilere uygun olduğu görülebilir. Her bir alt neden formülü genellikle, bir trafik ışığının, doluluk oranı yüksek bir yoldaki trafik akışını engellememesi gerektiğini ifade eder. 

 % \underline{oc}=\overline{oc}=1,\overline{p}=1, from 384 formulas
 %P 1 1 ( ( x5 = 0.0 ) & ( ( x3 > 10.0 ) & ( x5 = 0.0 ) ) )
%\underline{oc} \overline{oc} = 2, \overline{p}=1, from 12504 formulas
 %P 1 1 ( ( x5 = 0.0 ) & ( ( ( x4 > 10.0 ) & ( x6 = 0.0 ) ) | ( x3 > 10.0 ) ) )
 %( ( ( P 1 1 ( ( x6 = 0.0 ) & ( x4 > 10.0 ) ) ) | ( P 1 1 ( ( x5 = 0.0 ) & ( x3 > 10.0 ) ) ) ) | ( P 1 1 ( ( x5 = 1.0 ) & ( x0 > 30.0 ) ) ) ) | ( P 1 1 ( ( x6 = 1.0 ) & ( x1 > 30.0 ) ) )
 % Bir de heuristic ile calistirayim

 %Örnek şablon sayısı şu kadardır, biz ön eleme algoritmamız ile bunlardan sadece şu kadarını değerlendiriyoruz.
 \section{Sonuç}
 Bu bildiride, etiketli sinyal kümesi üzerinde ptSTL formülleri türetilerek istenmeyen olayların nedenlerini bulmaya yönelik bir \textit{kontrol edilebilir neden şablonu} sunulmuştur. Geliştirilen bu şablon, kontrol edilebilir tüm nedenleri kapsamaktadır. İstenmeyen olayların frekansını düşürmek için uygulanacak takip aşaması, bu nedenlerin bir kontrolcü ile engellenmesidir . Bildiride sunulan şablonun, daha önce kullanılan şablonlardan daha açıklayıcı olması, çözümün başarısını artıracaktır. Sistem-bağımsız bu istenmeyen olay tespit etme ve azaltma algoritmasının verimliliğini artırmak için atılabilecek bir sonraki adım, verilen şablondaki formüllerin \textit{kontrol edilemeyenlerinin} elenmesi için bulunacak daha verimli bir algoritma veya şablonda yapılacak bir basitleştirme olacaktır.
 \section*{B{\footnotesize İ}lg{\footnotesize İ}lend{\footnotesize İ}rme}
Bu çalışma 117E242 numaralı TÜBİTAK projesi ile desteklenmiştir.

%  \textit{Bilgilendirmeler Gizlenmiştir}

%\section*{TEŞEKKÜR}
%Yazarın teşekkür etmek istediği kurum yada kişiler burada belirtilecek.



%\bibliographystyle{IEEEtran}

\bibliographystyle{IEEEtran}
\bibliography{ebru_t}

\begin{thebibliography}{1}
\providecommand{\url}[1]{#1}
\csname url@samestyle\endcsname
\providecommand{\newblock}{\relax}
\providecommand{\bibinfo}[2]{#2}
\providecommand{\BIBentrySTDinterwordspacing}{\spaceskip=0pt\relax}
\providecommand{\BIBentryALTinterwordstretchfactor}{4}
\providecommand{\BIBentryALTinterwordspacing}{\spaceskip=\fontdimen2\font plus
\BIBentryALTinterwordstretchfactor\fontdimen3\font minus
  \fontdimen4\font\relax}
\providecommand{\BIBforeignlanguage}[2]{{%
\expandafter\ifx\csname l@#1\endcsname\relax
\typeout{** WARNING: IEEEtran.bst: No hyphenation pattern has been}%
\typeout{** loaded for the language `#1'. Using the pattern for}%
\typeout{** the default language instead.}%
\else
\language=\csname l@#1\endcsname
\fi
#2}}
\providecommand{\BIBdecl}{\relax}
\BIBdecl

\bibitem{7424283}
A.~{Lavin} and S.~{Ahmad}, ``Evaluating real-time anomaly detection algorithms
  -- the numenta anomaly benchmark,'' in \emph{2015 IEEE 14th International
  Conference on Machine Learning and Applications (ICMLA)}, Dec 2015, pp.
  38--44.

\bibitem{Kavraki:MPlanning}
A.~Bhatia, L.~E. Kavraki, and M.~Y. Vardi, ``Motion planning with hybrid
  dynamics and temporal goals,'' in \emph{{IEEE} Conference on Decision and
  Control}, Atlanta, GA, 2010, pp. 1108--1115.

\bibitem{tabuada2006linear}
P.~Tabuada and G.~J. Pappas, ``Linear time logic control of discrete-time
  linear systems,'' \emph{Automatic Control, IEEE Transactions on}, vol.~51,
  no.~12, pp. 1862--1877, 2006.

\bibitem{LaWaAnBe-ICRA10}
M.~Lahijanian, J.~Wasniewski, S.~B. Andersson, and C.~Belta, ``Motion planning
  and control from temporal logic specifications with probabilistic
  satisfaction guarantees,'' in \emph{{IEEE} International Conference on
  Robotics and Automation}, Anchorage, AK, 2010, pp. 3227--3232.

\bibitem{Yoo-RSS-17}
C.~Yoo and C.~Belta, ``Rich time series classification using temporal logic,''
  in \emph{Proceedings of Robotics: Science and Systems}, Cambridge,
  Massachusetts, July 2017.

\bibitem{mining_journal}
X.~Jin, A.~Donze, J.~V. Deshmukh, and S.~A. Seshia, ``Mining requirements from
  closed-loop control models,'' \emph{IEEE Transactions on Computer-Aided
  Design of Integrated Circuits and Systems}, vol.~34, no.~11, pp. 1704--1717,
  Nov 2015.

\bibitem{Bartocci2014}
E.~Bartocci, L.~Bortolussi, and G.~Sanguinetti, ``Data-driven statistical
  learning of temporal logic properties,'' in \emph{FORMATS 2014, LNCS, vol
  8711}.\hskip 1em plus 0.5em minus 0.4em\relax Springer International
  Publishing, 2014, pp. 23--37.

\bibitem{CDC2019}
I.~Sağlam and E.~A. Gol, ``Cause mining and controller synthesis with stl,''
  in \emph{58th IEEE Conference on Decision and Control (CDC)}, 2019, pp.
  4589--4594.

\bibitem{coogan2016traffic}
S.~Coogan, E.~A. Gol, M.~Arcak, and C.~Belta, ``Traffic network control from
  temporal logic specifications,'' \emph{IEEE Transactions on Control of
  Network Systems}, vol.~3, no.~2, pp. 162--172, 2016.

\end{thebibliography}

\end{document}